\documentclass[12pt]{article}

\def\pref{{\rm Pref}}
\def\suff{{\rm Suff}}
\def\fact{{\rm Fact}}
\def\subw{{\rm Subw}}

\usepackage{fullpage}
\usepackage{amsmath}
\usepackage{amsthm}
\usepackage{amssymb}
\usepackage{epsfig}
\usepackage{float}

\newtheorem{theorem}{Theorem}
\newtheorem{corollary}[theorem]{Corollary}
\newtheorem{lemma}[theorem]{Lemma}

\newtheorem{defin}[theorem]{Definition}

\newtheorem{examp}[theorem]{Example}

\newtheorem{rema}[theorem]{Remark}

\title{The computational complexity of universality problems for prefixes, suffixes, factors, and subwords of regular languages}

\author{Narad Rampersad\\
Department of Mathematics and Statistics\\
University of Winnipeg\\
515 Portage Avenue\\
Winnipeg, Manitoba R3B 2E9\\
Canada\\
{\tt narad.rampersad@gmail.com}\\
and\\ Jeffrey Shallit and Zhi Xu\\
School of Computer Science \\
University of Waterloo \\
Waterloo, Ontario N2L 3G1\\
Canada \\
{\tt shallit@cs.uwaterloo.ca}
}

\begin{document}

\maketitle

\begin{abstract}
In this paper we consider the computational complexity of the
following problems:  given a DFA or NFA representing a regular language
$L$ over a finite alphabet $\Sigma$,
is the set of all prefixes (resp., suffixes, factors, subwords) of
all words of $L$ equal to $\Sigma^*$?  In the case of testing universality
for factors of languages represented by DFA's, we find an interesting
connection to \v{C}ern\'y's conjecture on synchronizing words.
\end{abstract}

\section{Introduction}
The complexity of deciding universality --- i.e., whether a particular
formal language over a finite alphabet $\Sigma$ contains all of
$\Sigma^*$ --- is a recurring theme in formal language theory
\cite{Holzer}.  Frequently
it is the case that testing membership for a single word is easy, while
testing membership for all words simultaneously is hard.
For example, in two
classic papers,  Bar-Hillel, Perles, and Shamir proved
that testing universality for context-free languages represented by
grammars is recursively unsolvable \cite[Thm.\ 6.2 (a), p.\ 160]{Bar-Hillel},
and Meyer and Stockmeyer \cite[Lemma 2.3, p.\ 127]{Meyer} proved that
testing universality for regular languages represented by nondeterministic
finite automata is PSPACE-complete.   (Also see \cite{Hartmanis,Hunt}.)

Kozen \cite[Lemma 3.2.3, p.\ 261]{Kozen} proved that
determining whether the intersection of the languages accepted by $n$ DFA's 
is empty is PSPACE-complete. By complementing each DFA, we get

\begin{lemma}
      The following decision problem is PSPACE-complete:   \\
\centerline{\rm Given $n$ DFA's $M_1, M_2, \ldots, M_n$, each with input alphabet $\Sigma$,
	 is $\bigcup_{1 \leq i \leq n} L(M_i) = \Sigma^*$?}
\end{lemma}

Another frequently occurring theme is looking at prefixes, suffixes,
factors, and subwords of languages.  We say a word $y$ is a {\it factor} of
a word $w$ if there exists words $x,z$ such that $w = xyz$.   If in
addition $x = \epsilon$, the empty word, then we say $y$ is a {\it prefix}
of $w$; if $z = \epsilon$, we say $y$ is a {\it suffix}.  Finally, we
say $y$ is a {\it subword} of $w$ if we can write $y = a_1 a_2 \ldots a_n$
and $w = w_1 a_1 w_2 a_2 \cdots w_n a_n w_{n+1}$ for some letters
$a_i \in \Sigma$ and words $w_i \in \Sigma^*$.  (In the literature,
what we call factors are sometimes called ``subwords'' and what we
call subwords are sometimes called ``subsequences''.)

Let $L \subseteq \Sigma^*$ be a language.  We define 
$$\pref(L) = \lbrace x \in \Sigma^* \ : \ 
\text{there exists $y \in L$ such that $x$
is a prefix of $y$} \rbrace ,$$ 
and in a similar manner we define
$\suff(L)$, $\fact(L)$, and $\subw(L)$ for suffixes, factors, and
subwords.

In this paper we combine these two themes, and examine the computational
complexity of testing universality for the prefixes, suffixes, factors,
and subwords of a regular language.  As we will see, the complexity 
depends both on how the language is represented (say, by a DFA or NFA),
and on the particular type of factor or subword
involved.  In the case where we are testing universality for suffixes of
a language represented by a DFA, we find an interesting connection with
\v{C}ern\'y's celebrated conjecture on synchronizing words.

     Let us briefly mention some motivation for examining these
questions.  First, they are related to natural questions involving infinite
words.  By $\Sigma^{\omega}$ we mean the set of all right-infinite
words over $\Sigma$, that is, infinite words of the form
$a_0 a_1 a_2 \cdots$, where $a_i \in \Sigma$ for all integers $i \geq 0$.
Similarly, by ${}^\omega \Sigma$ we mean the set of all left-infinite
words over $\Sigma$, that is, infinite words of the form
$\cdots a_2 a_1 a_0$.  Finally, by ${}^\omega\Sigma^\omega$ we mean the
set of all (unpointed) bi-infinite words of the form
$\cdots a_{-2} a_{-1} a_0 a_1 a_2 \cdots$, where two words are
considered the same if one is a finite shift of the other.  

     Given a language of finite words $L \subseteq \Sigma^*$, we define
$L^\omega = \lbrace x_1 x_2 x_3 \cdots \ : \ x_i \in L-\lbrace \epsilon
\rbrace \rbrace$, the language of right-infinite words generated by $L$.
In a similar way we can define ${}^\omega L$ and ${}^\omega L^\omega$.

Given a finite set of finite words $S$, it is a natural question whether
all right-infinite words (resp., left-infinite words, bi-infinite
words) can be generated using only words of $S$.  The following
results are not difficult to prove using the usual argument from
K\"onig's infinity lemma or a compactness argument:

\begin{theorem}
Let $S \subseteq \Sigma^*$ be a finite set of finite words over the
finite alphabet $\Sigma$.  Then
\begin{itemize}
\item[(a)] $S^\omega = \Sigma^\omega$ iff $\pref(S^*) = \Sigma^*$.

\item[(b)] ${}^\omega S = {}^\omega \Sigma$ iff $\suff(S^*) = \Sigma^*$.

\item[(c)] ${}^\omega S^\omega = {}^\omega\Sigma^\omega$ iff
$\fact(S^*) = \Sigma^*$.
\end{itemize}
\end{theorem}

This theorem, then, leads naturally to the questions on prefixes,
suffixes, and factors considered in this paper.

Another motivation is the following:  as is well-known, the following
decision problem is recursively unsolvable \cite{Paterson}:

{\it Given a finite set of square matrices of the same dimension,
with integer entries, decide if some product of them is the all-zeros matrix.}

On the other hand, if the integer matrices are replaced by Boolean
matrices, and the multiplication is Boolean matrix multiplication,
the problem is evidently solvable, as there are only finitely many
different possibilities to consider.  We will show in
Corollary~\ref{boolean} below that the decision problem for Boolean matrices is 
PSPACE-complete.

\section{Basic observations}
\label{basic-sec}

     We recall some observations from
\cite{Kao1}.

     Given a DFA or NFA $M = (Q, \Sigma,\delta,q_0,F)$, we can easily construct
NFA's accepting $\pref(L(M))$, $\suff(L(M))$, $\fact(L(M))$, and
$\subw(L(M))$, as follows:

    To accept $\pref(L(M))$ with $M' = (Q, \Sigma, \delta, q_0, F')$,
 we simply change the set of final states as follows:  a state $q$ is in
$F'$ if and only if there is a path from $q$ to a state of $F$.  Note
that in this case, if $M$ is a DFA, then so is $M'$.

     To accept $\suff(L(M))$, we simply change the set of initial
states as follows:  a state $q$ is initial if and only if there is a
path from $q_0$ to $q$.  This construction creates a ``generalized''
NFA which differs from the standard definition of NFA in that it allows
an arbitrary set of initial states $I$, instead of just a single
initial state.  To get around this problem, we can simply create a new
initial state $q'_0$ and $\epsilon$-transitions to all the states of
$I$, and use the standard algorithm to get rid of the
$\epsilon$-transitions  without increasing the number of states
\cite{Hopcroft}.

     To accept $\fact(L(M))$, we do both of the transformations given
above.  In fact, there is an even simpler way to create a ``generalized
NFA'' accepting $\fact(L(M))$:  starting with $M$, remove all states
not reachable from the initial state, and remove all states from which
one cannot reach a final state.  Then make all the remaining states
both initial and final.

     To accept $\subw(L(M))$, we add $\epsilon$-transitions linking
every pair of states for which there is an ordinary transition.  
This produces an NFA-$\epsilon$, and again the $\epsilon$-transitions can
easily be removed without increasing the number of states.

\section{Universality for DFA's}

    In this section we assume that our regular language is represented
by a DFA $M = (Q, \Sigma, \delta, q_0, F)$.  We assume our DFA is
{\it complete}, that is, that $\delta:Q \times \Sigma \rightarrow \Sigma^*$
is well-defined for all elements of its domain.  

\subsection{Universality for prefixes}

     Of all our results, universality for $\pref(L)$ when $L$ is a DFA
is the easiest to decide.  By a well-known construction, given a
DFA $M$ accepting $L$, we can create a new DFA $M'$ as follows:  
$M' = (Q, \Sigma, \delta, q_0, Q-F')$, where $F' = 
\lbrace q \in Q \ : \ \text{there exists a path from $q$ to an element
of $F$} \rbrace$.
Then $L(M') = \overline{\pref(L(M))}$.
Furthermore, we can determine $F'$ in linear time
as follows: we reverse all arrows in $M$, add a new state $q'$
with an arrow to each final state in $F$, and determine which states
are reachable from $q'$.  The resulting set of states equals $F'$.
Now $\pref(L) = \Sigma^*$ if and only if $L(M') = \emptyset$, and this latter
condition can easily be tested by using depth-first search in $M'$,
starting from its initial state.  We have proved:

\begin{theorem}
    Given a DFA $M$ with input alphabet $\Sigma$,
we can test if $\pref(L(M)) = \Sigma^*$ in linear time.
\end{theorem}

\subsection{Universality for suffixes}
\label{suffix-dfa}

     Universality for suffixes is, perhaps surprisingly, much more
difficult.  

\begin{theorem}
     The decision problem\\
\centerline{\rm Given a DFA $M$ with input alphabet $\Sigma$, is $\suff(L(M)) = \Sigma^*$?} \\
is PSPACE-complete.
\end{theorem}

\begin{proof}
Suppose $M$ has $n$ states.  To see that the decision problem is in
PSPACE, note that by the results in section~\ref{basic-sec}, we can
convert $M$ to an NFA $M'$ accepting $\suff(L(M))$, having only one more
state than $M$.  As we noted above, the universality decision
problem for NFA's is in PSPACE.

%We can then convert $M'$ to a DFA $M''$ using the usual subset
%construction.  Finally, we change all of $M''$'s accepting states to
%non-accepting, and vice versa, resulting in a DFA $M'''$ that accepts
%$\emptyset$ if and only if $\suff(L(M)) = \Sigma^*$.  However, if $M'''$
%has $N$ states and
%accepts a word, it must accept a word of length $< N$.  It follows
%that if $M'$ rejects a word, it must reject a word of length
%$< 2^n$.  Therefore, the following is an algorithm running in
%nondeterministic polynomial space:  guess, symbol-by-symbol, a word $x$
%of at length $< 2^n$ and simulate $M'$ on it.  If $M'$ rejects $x$, then
%$\suff(L(M)) \not= \Sigma^*$.  We use a counter to make sure $|x| < 2^n$;
%this counter needs no more than $n$ bits.  By Savitch's theorem, there is
%deterministic polynomial space algorithm for the same problem.  

Now we show that the decision problem is PSPACE-hard.  To do so, we
reduce from the following well-known PSPACE-complete problem:
Given $n$ DFA's $M_0, M_1, \ldots, M_{n-1}$,
is there a word accepted by all of them?
More precisely, we reduce from the following problem:
given $n$ DFA's $M_0, M_1, \ldots, M_{n-1}$, is the union of all
their languages equal to $\Sigma^*$?  

Suppose $M_i = (Q_i, \Sigma, \delta_i, q_0^i, F_i)$ for $0 \leq i < n$.
Without loss of generality, we assume no $M_i$ has transitions into the
initial state; if this condition does not hold, we alter $M_i$ to add a new
initial state and transitions out of this initial state that coincide with
the original initial state.   Let $a, c$ be letters not in $\Sigma$, and
let $\Delta = \Sigma \ \cup \ \lbrace a,c \rbrace$.
We create a new DFA $M = (Q, \Delta, \delta, q, F)$ which is illustrated
in Figure~\ref{suff1} below. 

\begin{figure}[H]
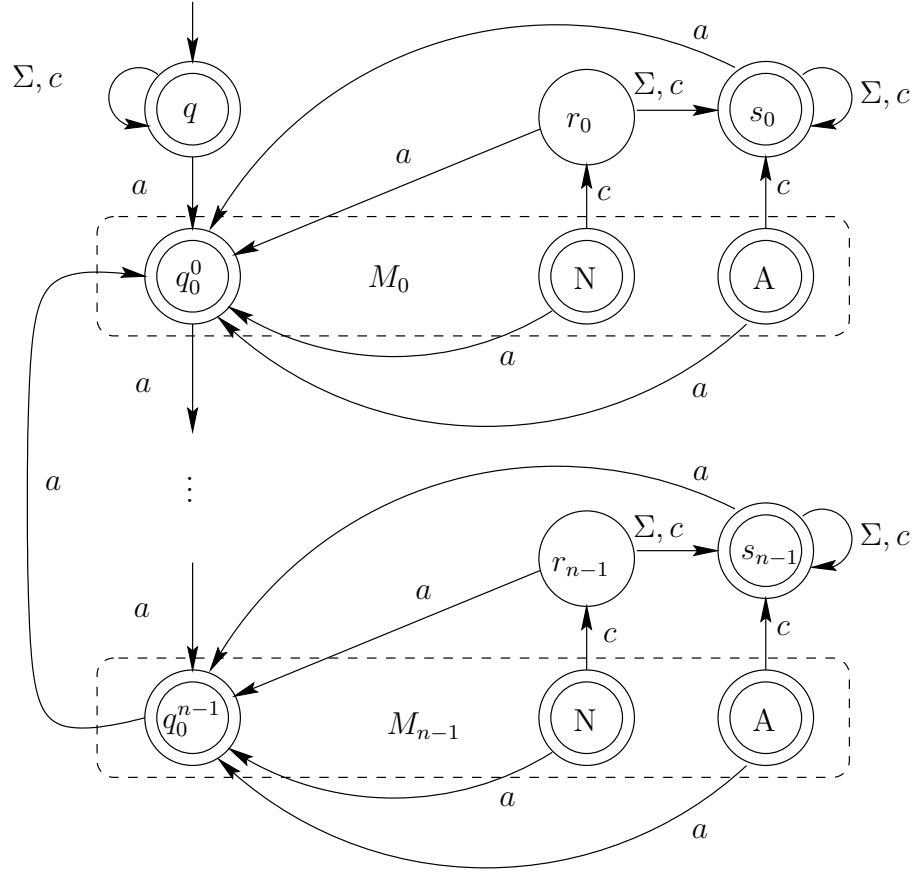

\begin{center}
\input ufig1.pstex_t
\end{center}
\caption{The reduction for suffixes}
\label{suff1}
\end{figure}

The idea of the construction is as follows:  our new machine $M$ incorporates
all the automata $M_0, M_1, \ldots, M_{n-1}$, but
we change all states of each
$M_i$ to final. For each $M_i$, we add two new states, $r_i$ (nonaccepting)
and $s_i$ (accepting).
Each formerly
nonaccepting state of $M_i$ has a transition on $c$ to $r_i$, and
each accepting state has a transition on $c$ to $s_i$; this is illustrated
in Figure~\ref{suff1} with the states labeled ``N'' (for nonaccepting) and
``A'' (for accepting).  Each state of $M_i$, other than the initial state,
has a transition on $a$ back to the initial state.

Each of the $M_i$ is linked via their initial states;  $q_0^i$ is linked
to $q_0^{(i+1) \bmod n}$ with a transition on $a$.  There are also
transitions on each letter in $\Sigma \cup \lbrace c \rbrace$
from both $r_i$ and $s_i$ to $s_i$.  There are also transitions on $a$
from both $r_i$ and $s_i$ to $q_0^i$.

The reader should check that $M$ is actually a complete DFA, and
furthermore $M$ accepts all words, except possibly some of
those that end in a word of the form $a x c$, where $x$ is rejected by
some $M_i$.  We now prove that $\suff(L(M)) = \Delta^*$ iff
$\bigcup_{0 \leq i < n} L(M_i) = \Sigma^*$.

Suppose $\suff(L(M)) = \Delta^*$.  Then every word in $\Delta^*$ is a suffix
of some word in $L(M)$.  In particular, this is true for every word of the
form $a w c$, where $w \in \Sigma^*$.   So $y a w c$ is accepted by $M$ for some
$y$ (depending on $w$).  However, every transition on $a$ leads to a 
state of the form $q_0^i$ for some $i$.  Transitions on elements of
$\Sigma$ then keep us inside the copy of $M_i$, and then processing
$c$ leads to either $r_i$ or $s_i$, depending on whether $M_i$ rejects or
accepts $w$, respectively. Since $yawc$ is accepted, this means that
we end in $s_i$, so $w$ is accepted by $M_i$.  Since $w$ was arbitrary,
this shows that every word is accepted by some $M_i$.

On the other other hand, suppose $\bigcup_{0 \leq i < n} L(M_i) = \Sigma^*$.
We need to show each $x \in \Delta^*$ is a suffix of some word accepted
by $M$.   If $x$ contains no $a$'s, then it is accepted by $M$ by a loop
on the initial state $q$.  Otherwise, we can write
$x = y a z$, where $z$ contains no $a$'s.  Then in processing $x$,
reading $a$ leads us to some state of the form $q_0^i$.  
If $z$ also contains no $c$'s, then processing $x$ in its entirety leads
to a state of $M_i$, all of which have been made accepting in our
construction.  Thus $x$ is accepted.  Otherwise, we can write 
$z = v c w$, where $v$ contains no $c$'s.  If $w$ is nonempty, then
processing $x$ leads to the state $s_i$, which is accepting, and so
$x$ is accepted.  Thus we may assume $w$ is empty and $z = vc$ for
some $v \in \Sigma^*$.  If reading $x = y a v c$ leads to $s_i$ for some $i$,
then $x$ is accepted by $M$.  Otherwise, reading $x$ leads to $r_i$.
By hypothesis $v$ is accepted by some $M_j$.
Let $s = (j-i) \bmod n$, and consider $a^s cc x$.  Then reading
$a^s cc$ leads to $s_{(j-i) \bmod n}$.  Hence reading
$a^s cc x$ leads to $s_j$, and it is accepted, and so 
$x \in \suff(L(M))$.  

\end{proof}

\subsection{Universality for factors}

\begin{theorem}
The decision problem \\
	\centerline{\rm Given a DFA $M$ with input alphabet $\Sigma$, is $\fact(L(M)) = \Sigma^*$?} \\
is solvable in polynomial-time.
\end{theorem}

\begin{proof}
Terminology: we say a state $q$ is {\it dead} if no accepting state can be
reached from $q$ via a possibly empty path.  If a DFA has a dead state $d$
then every state reachable from it is also dead, so there is an
equivalent DFA with only one dead state and all transitions from that dead
state lead to itself.

We say a state $r$ is
{\it universal} if no dead state is reachable from it via a possibly empty
path.   We say a state is {\it reachable} if there is some path to it from
the start state.  We say a DFA is {\it initially connected} if all states are
reachable.  A DFA $M = (Q, \Sigma, \delta, q_0, F)$
has a {\it synchronizing word} $w$ if
$\delta(p, w) = \delta(q,w)$ for all states $p, q$.

We need two lemmas.

\begin{lemma}
If a DFA $M$ has a reachable universal state, then $\fact(L(M)) =
\Sigma^*$.
\end{lemma}

\begin{proof}
Let $M = (Q, \Sigma, \delta, q_0, F)$.
Let $q$ be a reachable universal state, and let $x$ be such that
$\delta(q_0, x) = q$.  Consider any word $y$, and let $\delta(q, y) = r$.
Then no dead state is reachable from $r$, for otherwise it would be
reachable from $q$.  So there exists a word $z$ such that $\delta(y, z) =
s$, and $s$ is an accepting state.  Then $\delta(q_0, xyz) = s$, so $xyz$ is
accepted, and hence $y \in \fact(L(M))$.  But $y$ was arbitrary, so
$\fact(L(M)) = \Sigma^*$.
\end{proof}

\begin{lemma}
Suppose the DFA $M$ is initially connected, has no universal
states, and has exactly one dead state.  Then there exists $x\not\in
\fact(L(M))$ if and only if there is a synchronizing word for $M$.
\label{second}
\end{lemma}

\begin{proof}
Suppose $M$ has a synchronizing word $x$. Then there exists a state $q$
such that for all all states $p$ we have $\delta(p,x) = q$.  Since, as
noted above, all transitions from the unique dead state $d$ must go to itself,
we must have $q = d$.
Then for all states $p$ we have $\delta(p,x) = d$.  So $x$ is not
in $\fact(L(M))$, because every path labeled $x$ goes to a state from which
one cannot reach a final state.

Now suppose there is $x \not\in \fact(L(M))$.  Then for all $y, z$ we have
$yxz \not\in L(M)$.  In other words, no matter what state we start in, $xz$
leads to a nonaccepting state.  Then no matter what state we start in x
leads to a state from which no accepting state can be reached.  But
there is only one such state, the dead state $d$.  So it must be the case
that $x$ always leads to $d$, and so $x$ is a synchronizing word.
\end{proof}

We can now prove the theorem.
The following algorithm decides whether
$\fact(L(M)) = \Sigma^*$ in polynomial time:

\begin{enumerate}
\item Remove all states not reachable from the start state by a (possibly
empty) directed path.

\item Identify all dead states via depth-first search.  If $M$ has at least
one dead state, modify $M$ to replace all dead states with a single dead
state $d$.

\item Identify all universal states via depth-first search.  If there is
a universal state, answer ``Yes" and halt.

\item Using the polynomial-time
procedure mentioned in Volkov's survey \cite{Volkov},
decide if there is a synchronizing
word.  If there is, answer ``No"; otherwise answer ``Yes".
\end{enumerate}

To see that it works, we already observed that we can replace
all dead states by a single dead state 
without changing the language accepted by $M$. Furthermore,
if a DFA has no universal states, then it has at least one dead state
(for otherwise every state would be universal).    So when we reach
step 4 of the algorithm, we are guaranteed that $M$ has exactly one dead
state and we can apply Lemma~\ref{second}.
\end{proof}

\subsection{Universality for subwords}

This is covered in section~\ref{subword} below.

\section{Universality for NFA's}

In this section we consider the same problems as before, but now we
represent our regular language by an NFA.  Some of these results
essentially appeared in \cite{Kao1}, but with different proofs and some
in weaker form.

\subsection{Universality for prefixes}

\begin{theorem}
The decision problem \\
\centerline{\rm Given an NFA $M$ with input alphabet $\Sigma$, is $\pref(L(M)) = \Sigma^*$?} \\
is PSPACE-complete.
\end{theorem}

\begin{proof}
In fact, this decision problem
is even PSPACE-complete when $M$ is restricted to be of the
form $A^R$, where $A$ is a DFA.   To see this, note that our construction
for suffix universality for DFA's given above, when reversed, gives
an NFA $M$ with the property that $\pref(L(M)) = \Sigma^*$ if and only
if $\bigcup_{0 \leq i < n} L(M_i) = \Sigma^*$.  
\end{proof}

\subsection{Universality for suffixes}

Already handled in section~\ref{suffix-dfa}.

\subsection{Universality for factors}
\label{factor-nfa}

Although, as we have seen, universality for $\fact(L(M))$ is testable in
polynomial-time when $M$ is a DFA, the same decision problem becomes 
PSPACE-complete when $M$ is an NFA.  To see this, we again reduce
from the universality problem for $n$ DFA's.  Figure~\ref{fact2} illustrates
the construction.  Given the DFA's $M_0, M_1, \ldots, M_{n-1}$, each
with input alphabet $\Sigma$, we
create a new NFA as illustrated.  We assume that $\Sigma$ does not contain
the letters $a, c$ and set $\Delta := \Sigma \ \bigcup \ \lbrace a, c \rbrace$.
Restricting our attention to the states
$q, r, s$ we get an NFA that accepts all words not having a word of the
form $a \Sigma^* c$ as a factor.  On the other hand, a word 
of the form $a w c$ for $w \in \Sigma^*$ is a factor of a word in $L(M)$ 
iff $w \in L(M_i)$ for some $i$.  We now claim that 
$\fact(L(M)) = \Delta^*$ iff $\bigcup_{0 \leq i < n} L(M_i) = \Sigma^*$.

    Suppose $\fact(L(M)) = \Delta^*$.  Then in particular
every factor of the form $a w b$, with $w \in \Sigma^*$ is a factor
of a word of $M$.  But the only way such a word can be a factor
is by entering one of the $M_i$ components on $a$ and exiting on $c$,
and there are only transitions on $c$ on states that were originally
final in $M_i$.  So $w$ must be accepted by some $M_i$.  Since $w$ was
arbitrary, we have $\bigcup_{0 \leq i < n} L(M_i) = \Sigma^*$.

    On the other hand, suppose $\bigcup_{0 \leq i < n} L(M_i) = \Sigma^*$.
We claim every word $x$ in $\Delta^*$ is in $\fact(L(M))$.  To see this,
note that if $x$ contains no subword of the form $a w c$, with $c \in \Sigma^*$,
then it is accepted by a path starting from state $q$ and only involving
the states $q, r, $ and $s$.  Otherwise $x$ contains a subword of
form $awc$.  Identify all the positions of $c$'s in $x$ and write
$x = x_1 c x_2 c \cdots x_{n-1} c x_n$, where each $x_i \in \Sigma \ \cup
\ \lbrace a \rbrace$.  
If an $x_i$ contains no $a$'s,
there is a path from $q$ to $q$ labeled $x_i c$. 
Otherwise $x_i$ contains at least one $a$.  Identify the position of
the last $a$ in $x_i$, and write $x_i = y_i a z_i$, 
where $z_i$ contains no $a$'s.   Then starting in $q$ and reading $y_i$
takes us to either state $q, r,$ or $s$; reading the $a$ takes us to
$t$ and then to any $q_0^j$.  Since $z_i \in \Sigma^*$, and since
$\bigcup_{0 \leq i < n} L(M_i) = \Sigma^*$, we can choose the particular
$M_j$ that accepts $z_i$.  Then reading $c$ takes us back to state $q$.
By this argument we see that $xc$ is always accepted by $M$, and hence
$x$ is a factor of $L(M)$.  

\begin{figure}[H]
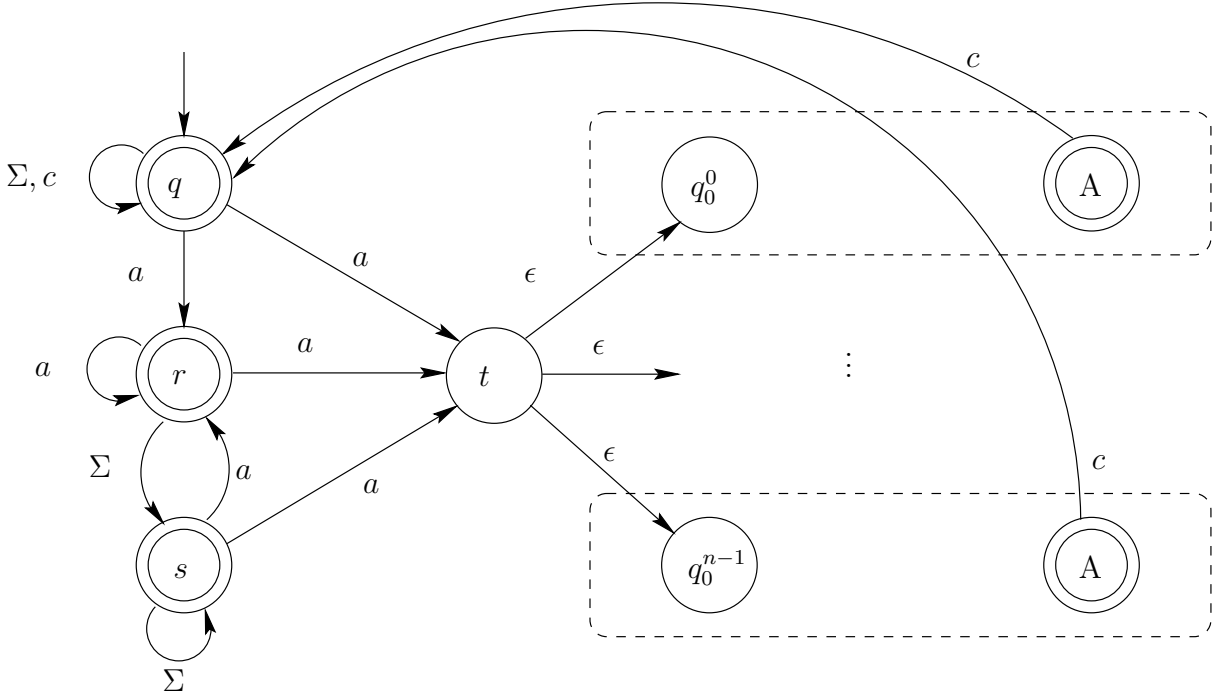

\begin{center}
\input univ2.pstex_t
\end{center}
\caption{The reduction for factors}
\label{fact2}
\end{figure}

\begin{theorem}
The decision problem\\
\centerline{\rm Given an NFA $M$ with input alphabet $\Sigma$, is $\fact(L(M)) = \Sigma^*$?}\\
is PSPACE-complete.
\end{theorem}

As we mentioned in the introduction, this result has an interesting
interpretation in terms of Boolean matrices.  Given an NFA $M = (Q,\Sigma,\delta,q_0,F)$,
we can form $|\Sigma|$ different matrices $M_a$, for each $a \in \Sigma$,
as follows:  $M_a$ has a $1$ in row $i$ and column $j$ 
if $q_j \in \delta(q_i, a)$, and a $0$ otherwise.  Then it is easy
to see that for all words $w = c_1 c_2 \cdots c_k$, that
$M_w := M_{c_1} M_{c_2} \cdots M_{c_k}$ has a $1$ in row $i$ and column
$j$ iff $q_j \in \delta(q_i, w)$.  

Assume that $M$ is an NFA in which every state is reachable from the
start state and that a final state can be reached from every state.
(If $M$ does not fulfill these conditions, we can simply delete
the appropriate states.)  Then form $M_a$ for each $a \in \Sigma$.
We claim that some product of the $M_a$
equals the all-zeros matrix iff $\fact(L(M)) \not= \Sigma^*$.  
For suppose there is some product, say $M_y$ for $y = c_1 \cdots c_k$,
that equals the all-zeros matrix.  Then no matter what state we start in,
reading $y$ takes us to no state, so $xyz$ is rejected for all $x, z$.
Hence $y \not\in \fact(L(M))$.  On the other hand, if 
$\fact(L(M)) \not= \Sigma^*$, then there must be some
$y \not\in \fact(L(M))$.  We claim $M_y$ is the all-zeros matrix.
If not, there exist $i,j$ such that $M_y$ has a $1$ in row $i$ and
column $j$.  Then since every state is reachable from the start
state, there exists $x$ such that $\delta(q_0, x) = q_i$.
Since a final state can be reached from every state, there exists
$z$ such that $\delta(q_j, z) \in F$.  Then $\delta(q_0,xyz) \in F$,
so $M$ accepts $xyz$ and $y \in \fact(L(M))$, contradicting our assumption.

We have therefore shown

\begin{corollary}
The decision problem\\
{\rm Given a finite list of square Boolean matrices of the same dimension,
is some product equal to the all-zeros matrix?} \\
is PSPACE-complete.
\label{boolean}
\end{corollary}

\subsection{Universality for subwords}
\label{subword}

     We now consider the problem of determining, given an NFA $M$,
whether $\subw(L(M)) = \Sigma^*$.  

\begin{lemma}
Let $M = (Q,\Sigma, \delta, q_0, F)$ be an NFA such that (a)
every state is reachable from $q_0$ and (b) a final state is reachable
from every state.  Then
$\subw(L(M)) = \Sigma^*$ if and only if
the transition diagram of $M$ has a strongly connected component $C$
such that, for each letter $a \in \Sigma$, there are two states of
$C$ connected by an edge labeled $a$.
\label{subword-lemma}
\end{lemma}

\begin{proof}
Suppose the transition diagram of $M$ has a reachable strongly connected
component $C$ with the given property.  Then to obtain any word $w$ as a subword
of a word in $L(M)$, use a word to enter the strongly connected component $C$,
and then travel successively to states of $C$ where there is an arrow out 
labeled with each successive letter of $w$.  Finally, travel to a final
state.  

For the converse, assume $\subw(L(M)) = \Sigma^*$, but the transition
diagram of $M$ has no strongly connected component with the given
property.  Then since any directed graph can be decomposed into a
directed acyclic graph on its strongly connected components,
we can write any $w \in L(M)$ as $x_1 y_1 x_2 y_2 \cdots x_n$, where
$x_i$ is the word traversed inside a strongly connected component,
and $y_i$ is the letter on an edge linking two strongly connected
components.  Furthermore, $n \leq N$, where $N$ is the total
number of strongly connected components.
If $\Sigma = \lbrace a_1, a_2, \ldots, a_k \rbrace$,
then $\subw(L(M))$ omits the word
$w = (a_1 a_2 \cdots a_k)^{N+1}$,
because the first component encountered has no transition on
some letter $a_i$, so reading $a_1 a_2 \cdots a_k$ either
forces a transition to (at least) the next component of the DAG, or in the
case of an NFA, ends the computational path with no move.
Since there are only $N$ strongly connected components, we cannot
have $w$ as a subword of any accepted word.
\end{proof}
     
We can now prove

\begin{theorem}
Given an NFA $M$ with input alphabet $\Sigma$, we can determine if $\subw(L(M)) = \Sigma^*$ in linear
time.
\end{theorem}

\begin{proof}
First, use depth-first search to remove all states not reachable from the
start state.  Next, use depth-first search (on the transition diagram of
$M$ with arrows reversed) to remove all states from which one cannot
reach a final state.  Next, determine the strongly connected components
of the transition diagram of $M$ (which can be done in linear time 
\cite{Tarjan}).
Finally, examine all the edges of each strongly connected component $C$
to see if for all $a \in \Sigma$, there is an edge labeled $a$.
\end{proof}

\section{Shortest counterexamples}

We now turn to the following question:  given that
$\pref(L(M)) \not= \Sigma^*$, what is the length of the shortest
word in $\overline{\pref(L(M))}$, as a function of the number of
states of $M$?  We can ask the same question for suffixes, factors,
and subwords.

\begin{theorem}
Let $M$ be a DFA or NFA with $n$ states.
Suppose $\pref(L(M)) \not= \Sigma^*$.  Then the shortest word
in $\overline{\pref(L(M))}$ is

\begin{itemize} 
\item[(a)] of length $\leq n-1$ if $M$ is a DFA, and there exist
examples achieving $n-1$;

\item[(b)] of length $\leq 2^n$ if $M$ is an NFA, and there exist
examples achieving $2^{cn}$ for some constant $c$.
\end{itemize}

\label{pref-shortest}
\end{theorem}

\begin{proof}

\begin{itemize}

\item[(a)]  If $M$ is a DFA with $n$ states, our construction shows
$\overline{\pref(L(M))}$ can be accepted by a DFA $M'$ with $n$ states.  If
$M'$ accepts a string, it accepts one of length $\leq n-1$.  

An example achieving this bound is $L = a^{n-2}$, which can be
accepted by an $n$-state DFA, and the shortest string not in
$\pref(L)$ is $a^{n-1}$.  

\item[(b)]  The upper bound is trivial (convert the NFA for $M$ to
one for $\pref(L(M))$; then convert the NFA to a DFA and change
accepting states to non-accepting and vice versa; such a DFA has at
most $2^n$ states).   

The examples achieving $2^{cn}$ for some
constant $c$ can
be constructed using an idea in \cite{Kao1}:  there the authors construct
an $n$-state NFA $M$ with all states final such that the shortest string
not accepted is of length $2^{cn}$.  However, if all states are final,
then $\pref(L(M)) = L(M)$, so this construction provides the needed
example.
\end{itemize}
\end{proof}

\begin{theorem}
Let $M$ be a DFA or NFA with $n$ states.
Suppose $\suff(L(M)) \not= \Sigma^*$.  Then the shortest word
in $\overline{\suff(L(M))}$ is
of length $\leq 2^n$. There exist DFA's
achieving $e^{\sqrt{ c n \log n (1 + o(1))}}$ for a 
constant $c$, and there exist NFA's achieving
$2^{dn}$ for some constant $d$.
\end{theorem}

\begin{proof}
The upper bound of $2^n$ is just like in the proof of
Theorem~\ref{pref-shortest}.  The example for DFA's 
achieving $e^{c \sqrt{ n \log n (1 + o(1))}}$ for some
constant $c$ can
be constructed by using the construction in section~\ref{suffix-dfa},
with each $M_i$ a unary DFA accepting $b^{p_i-1} (b^{p_i})^*$ for
primes $p_1 =2$, $p_2 = 3$, etc.  The construction generates an automaton
of $O(p_1 + p_2 + \cdots p_n)$ states, and
the shortest word omitted as a suffix
is of length $\geq p_1 p_2 \cdots p_n$.

For NFA's, we take the construction in the proof of
Theorem~\ref{pref-shortest} (b) and construct the NFA for the
reversed language.  This can be done by reversing the order of each
transition, changing the initial state to final and all final
states to initial.  This creates a ``generalized NFA'' with a set
of initial states, but this can easily be simulated by an ordinary
NFA by adding a new initial state, adding $\epsilon$-transitions to the
former final states, and then removing $\epsilon$-transitions using the
usual algorithm.  This gives an example achieving $2^{dn}$ for some
constant $d$.
\end{proof}

\begin{theorem}
Let $M$ be a DFA or NFA with $n$ states.
Suppose $\fact(L(M)) \not= \Sigma^*$.  Then the shortest word
in $\overline{\fact(L(M))}$ is

\begin{itemize} 
\item[(a)] of length $O(n^2)$ if $M$ is a DFA, and there exist
examples achieving $\Omega(n^2)$;

\item[(b)] of length $\leq 2^n$ if $M$ is an NFA, and there exist
examples achieving $2^{cn}$ for some constant $c$.
\end{itemize}

\label{fact-shortest}
\end{theorem}

\begin{proof}

\begin{itemize}
\item[(a)]  The bounds come from known results on synchronizing
words \cite{Volkov,Rystsov}.

\item[(b)] The upper bound is clear.
For an example achieving $2^{cn}$, we use a construction from
\cite{Kao1}.
There the authors construct a ``generalized'' NFA $M$ of $n$ states
with all states both initial and final, such that the shortest
string not accepted is of length $2^{cn}$.  Such an NFA can be converted
to an ordinary NFA, as we have mentioned previously, at a cost of increasing
the number of states by $1$.  But for such an NFA, clearly
$\fact(L(M)) = L(M)$, so the result follows.
\end{itemize}
\end{proof}

We now turn to subwords.

\begin{theorem}
Given a DFA or NFA $M$ of $n$ states, with input alphabet $\Sigma$,
if $\subw(L(M)) \not= \Sigma^*$,
then the shortest word in $\overline{\subw(L(M))}$ is of length
at most $n+1$, and there exist examples achieving $n$.
\end{theorem}

\begin{proof}
The upper bound is implied by our proof of Lemma~\ref{subword-lemma}.
An example is provided by choosing an alphabet of $n$ symbols,
say $a_0, a_1,\ldots, a_{n-1}$
and constructing an NFA $M$ with $n+1$ states, say $q_0, q_1, \ldots, q_n$,
where $q_n$ is accepting and all other states are nonaccepting,
such that there is a loop on state $q_i$ on all symbols except
$a_i$, for $0 \leq i < n $.  Also, there is a transition
from $q_i$ to $q_{i+1}$ labeled $a_i$. 
Then $a_0 a_1 \cdots a_{n-1} a_0$ is not a subword of any word accepted
by $M$.
\end{proof}

\section{Sets of finite words}

As we mentioned in the introduction, one motivation for this work
were the problems of testing
if (a) $S^\omega = \Sigma^\omega$,
(b) ${}^\omega S  = {}^\omega \Sigma$, or
(c) ${}^\omega S^\omega  = {}^\omega \Sigma^\omega$ for a finite set of
words $S$.  However, our results thus far do not really resolve the
worst-case complexity of these questions, for two reasons.  First, as
we have seen, answering (a) involves testing if
$\pref(S^*) = \Sigma^*$ (and similarly for (b), (c)), which means that
to use our results, we must first construct a DFA or NFA for $S^*$.
While constructing
a linear-size NFA for $S^*$ is computationally easy,
we have no fast algorithm
for answering our questions in that case (although there clearly
are exponential-time algorithms).  On the other hand, there
are examples known where the smallest DFA for $S^*$ is exponentially
large in the size of $S$ (see \cite{Kao2}),
so our polynomial-time algorithm for prefixes
and factors does not give an algorithm running in polynomial time in the
size of $S$.

For prefixes and suffixes (cases (a) and (b) above), we can nevertheless
obtain an efficient algorithm.  We state the result for prefixes
only; the corresponding result for suffixes can be obtained by reversing
each word in $S$.

\begin{theorem}
We can test in linear time whether a finite set of finite words $S$
has the property that $\pref(S^*) = \Sigma^*$.
\end{theorem}

\begin{proof}
Let $k = |\Sigma|$.
The following algorithm suffices:  construct a trie from the words of
$S$, inserting each word successively.  If at any point we attempt to
insert a word $w$ such that some already-inserted word $x$ is a prefix
of $w$, do not insert $w$.  Similarly, if at any point we attempt to
insert a word $w$ that is a prefix of an already-inserted word $x$,
remove $x$ and insert $w$ instead.  Then $\pref(S^*) = \Sigma^*$ if
and only if every node in the trie has degree $0$ or $k$.
\end{proof}

The problem of the complexity of determining, given a finite set of
finite words $S \subseteq \Sigma^*$, whether $\fact(S^*) = \Sigma^*$,
is still open.

We can also address the question of the shortest word not
in $\fact(S^*$), given that $\fact(S^*) \not= \Sigma^*$.  

\begin{theorem}
For each $n \geq 1$ there exists a set of finite words of length
$\leq n$, such that the shortest word not in $\fact(S^*)$ is of
length $n^2 + n - 1$.
\label{shortest}
\end{theorem}

\begin{proof}
Let $S = \Sigma^n - \lbrace 0^{n-1} 1 \rbrace$.  Then it is easy to
verify that the shortest word not in $\fact(S^*)$ is
$0^{n-1} 1 (0^n 1)^{n-1}$.
\end{proof}

\section{Afterword}

After this research was completed,
we we learned that some of the same questions in our paper were
recently and independently addressed in an unpublished paper
of Pribavkina \cite{Pribavkina}.
In particular, she obtained a result similar to our Lemma~\ref{second}, and
a result more general than our Theorem~\ref{shortest}.

\section{Acknowledgment}

We are very grateful to Mikhail Volkov for letting us know about the paper
of Pribavkina, and to Elena Pribavkina for sending us a copy of her paper
in English.

\end{document}